\documentclass[authoryear,3p]{elsarticle}

\usepackage[utf8]{inputenc}
\usepackage{amsmath}
\usepackage{amsfonts}
\usepackage{amssymb}
\usepackage{amsfonts}
\usepackage{amsthm}
\usepackage{amscd}
\usepackage{color}
\usepackage{graphics}
\usepackage{graphicx}
\usepackage[toc,page]{appendix}
\usepackage{todonotes}

\usepackage[all]{xy}

\usepackage{algorithm,algorithmic}

\usepackage{tikz}

\usepackage{natbib}

\newtheorem{theorem}{\bf Theorem}
\newtheorem{lemma}[theorem]{\bf Lemma}
\newtheorem{proposition}[theorem]{\bf Proposition}

\newtheorem{remark}[theorem]{\bf Remark}

\newtheorem{example}[theorem]{\bf Example}

\newcommand{\av}[1]{\ensuremath{#1}}

\newcommand{\f}[1]{\mathbf{#1}}

\newcommand{\Q}{\mathbb{Q}}
\newcommand{\C}{\mathbb{C}}
\newcommand{\R}{\mathbb{R}}
\newcommand{\K}{\mathbb{K}}
\newcommand{\PP}{\mathbb{P}}
\newcommand{\p}{\mathbb{P}}
\newcommand{\E}{\mathbb{E}}
\newcommand{\A}{\mathbb{A}}
\newcommand{\Z}{\mathbb{Z}}
\newcommand{\G}{\mathbb{G}}

\newcommand{\im}{\mathrm{Im}\,}
\newcommand{\re}{\mathrm{Re}\,}
\newcommand{\aut}{\mathrm{Aut}}

\newcommand{\aff}{\mathrm{Aff}}

\newcommand{\simm}{\mathrm{Sim}}

\newcommand{\GL}{GL}

\newcommand{\I}{\mathrm{i}}
\newcommand{\e}{\mathrm{e}}

\begin{document}
\sloppy

\begin{frontmatter}

\title{Computing projective equivalences of special algebraic varieties}

\author[plzen2]{Michal Bizzarri}
\ead{bizzarri@ntis.zcu.cz}

\author[plzen1,plzen2]{Miroslav L\'{a}vi\v{c}ka}
\ead{lavicka@kma.zcu.cz}

\author[plzen1,plzen2]{Jan Vr\v{s}ek\corref{cor1}}
\cortext[cor1]{Corresponding author}
\ead{vrsekjan@kma.zcu.cz}

\address[plzen1]{Department of Mathematics, Faculty of Applied Sciences, University of West Bohemia,
         Univerzitn\'i~8,~306~14~Plze\v{n},~Czech~Republic}

\address[plzen2]{NTIS -- New Technologies for the Information Society, Faculty of Applied Sciences, University of West Bohemia, Univerzitn\'i 8, 306 14 Plze\v{n}, Czech~Republic}

\begin{abstract}
This paper is devoted to the investigation of selected situations when the computation of projective (and other) equivalences of  algebraic varieties  can be efficiently solved with the help of finding projective equivalences of finite sets on the projective line.  In particular, we design a unifying approach that finds for two  algebraic varieties $X,Y$ from special classes an associated set of automorphisms of the projective line (the so called good candidate set) consisting of candidates for the construction of  possible mappings $X\rightarrow Y$. The functionality of the designed method is presented on computing  projective equivalences of rational curves, on determining projective equivalences of rational ruled surfaces, on the detection of affine transformations between planar curves, and on computing similarities between two implicitly given algebraic surfaces. When possible, symmetries of given shapes are also discussed as special cases.

\end{abstract}

\begin{keyword}
Projective transformation, symmetry, rational curve, rational ruled surface, algebraic surface
\end{keyword}

\end{frontmatter}

%%%%%%%%%%%%%%%%%%%%%%%%%%%%%%%%%%%%%%%%%%%%%%%%%%%%%%%%%%%%%%%%%%%%%%%%%%%%%%%%%%%%%%%%%%%%%%%%%%
%%%%%%%%%%%%%%%%%%%%%%%%%%%%%%%%%%%%%%%%%%%%%%%%%%%%%%%%%%%%%%%%%%%%%%%%%%%%%%%%%%%%%%%%%%%%%%%%%%

%\begin{abstract}
%xxx xxx xxx\\

%\medskip\noindent
%\emph{Key words}:
%\end{abstract}

%
%  Main text of the article.
%
\section{Introduction}

Projective (or affine, or similar, or isometric) equivalencies and symmetries of geometric shapes is a fundamental concept in nature, science, engineering, architecture, etc. For instance, symmetries in the natural world has inspired people to integrate symmetry when designing tools, buildings, or artwork. Many biochemical processes are governed by symmetries. Hence, detecting a suitable class of equivalencies in  given geometric data is a problem in geometry processing that has attracted attention of researchers from different scientific areas for many years.  Numerous algorithms have been formulated to identify, extract, encode, and apply geometric equivalencies or symmetries and numerous applications immediately benefit from gained equivalency or symmetry information. Nowadays, geometric equivalencies and symmetries play a significant role also  in computer graphics, computer vision, or in pattern recognition.

In short, the main goal is to decide whether two given geometric shapes are related by a suitable transformation and in the affirmative case to detect all such equivalences. In many applications, it is sufficiently enough to find approximate equivalencies and symmetries of the given shapes, only. For this, one can identify mainly the following two practical reasons  -- first, the input shape is approximate (it is an simplified model of some real object), or second, computations cannot be provided exactly (solving of complicated systems of non-linear equations). However, one question still remains, i.e., how to solve the problem exactly at least for some special algebraic varieties.

This problem has become an active research area especially in recent years and one can find several papers focused on the detection and computation of symmetries and some equivalences of curves, see e.g. \cite{HuCo94,BrKn04,LeRG08,Le09}, or series of papers  \cite{Al14,AlHeMu14,AlHeMu14b,AlHeMu15,AlHeMu18}. The first paper devoted to the broadest group, i.e., to the general projective equivalences, has appeared quite recently, see \cite{HaJu18}. In this paper, the authors study equivalences of curves with respect to the projective group in an arbitrary space dimension. The formulated symbolic-numerical algorithm (based on  Gr\"{o}bner bases computation) is universal and provides good computational results for all presented examples with coefficients from $\mathbb{Q}$.  Later, the problem of deterministically computing the symmetries of a planar implicitly given  curve and the problem of deterministically checking whether or not two implicitly given, planar algebraic curves are similar, i.e., equal up to a similarity transformation, was investigated in \cite{AlLaVr18}. Nonetheless, solving this problem for surfaces in 3-space, at least for special classes, remains still an open question which deserves further research.

The paper is organized as follows. Section~$2$ recalls some basic facts concerning projective (and other) equivalences, finite rotation groups and  Grassmannians.  As the paper is focused on studying suitable situations that can be reduced to the computation of equivalences of finite sets on the projective line, we present in Section~3 two algorithms devoted to the detection of equivalences of finite point sets using cross-ratios and to the detection of equivalences of finite sets given by a polynomial relation. The formulated computational method is then used in  Section~$4$ on computing projective (and other) equivalences of several types of algebraic varieties, in particular on projective equivalences of rational curves, on projective equivalences of rational ruled surfaces, on the detection of affine transformations mapping a planar curve to another planar curve and on computing similarities between two implicitly given algebraic surfaces. The functionality of the designed unifying approach  is documented  on several examples. Finally, we conclude the paper in Section~5. 

\section{Preliminaries}\label{sec:prelim}

First we recall some fundamental facts and notions  whose knowledge is anticipated in the following sections.

\subsection{Projective transformations}

Recall that the \emph{projective space} $\PP^n_{\K}$ of dimension $n$ over the field $\K$ is the set of all lines through the origin in $\K^{n+1}$. It can be interpreted as the quotient $\K^{n+1} \setminus \{\f 0\} / \! \sim$, where $\sim$ denotes the equivalence relation of points lying on the same line going through the origin:
\begin{equation}
(x_0,\ldots,x_n) \sim (y_0,\ldots,y_n) \quad \mathrm{iff} \quad (x_0,\ldots,x_n) = \lambda (y_0,\ldots,y_n) \quad \mathrm{for} \; \mathrm{some} \; \lambda \in \K^*,
\end{equation}
where $\K^* = \K \setminus \{0\}$. Hence a point  in the projective space can be considered as an equivalence class in $\K^{n+1}$ split by $\sim$. To distinguish the class from its representative we will write colons between the coordinates and the square brackets instead of the round brackets, i.e., $\f x=[x_0:x_1:\cdots:x_n]$. Throughout the paper we will work mainly over the field of complex numbers, in which case we will write $\p^n$ instead of $\p^n_\C$.
%Solving a homogeneous linear equation
%\begin{equation}\label{hyper}
%u_0x_0+\ldots +u_nx_n=\f u \cdot \f x = 0,
%\end{equation}
%where `$\cdot$' is the canonical scalar product, we arrive at a $(n-1)$-dimensional subspace of $\PP^n$, i.e., a~\emph{hyperplane} with the homogeneous coordinates $\f u=[u_0:u_1:\cdots:u_n]$. As the incidence condition \eqref{hyper} is completely symmetric, it remains unchanged by interchanging the point and hyperplane coordinates. We say that the points and hyperplanes are \emph{dual} to each other and the space $(\PP^n)^\vee$ whose points are identified with the hyperplanes of $\PP^n$ is called a \emph{dual projective space}.

Fixing a hyperplane $\omega: x_0=0$ as a hyperplane at infinity, or an \emph{ideal hyperplane}, we obtain the \emph{affine space} $\A^n_\K$ embedded into $\PP^n_\K$ via $(x_1,\ldots,x_n)\mapsto[1:x_1:\cdots:x_n]$. Conversely a point  $\f x=[x_0:\cdots:x_n] \in \PP^n \setminus \omega=\A^n$ has the affine coordinates $x=(x_1/x_0,\ldots,x_n/x_0)$. The oval quadric $\Omega:\,  x_0=x_1^2+\ldots+x_n^2=0$ lying in $\omega\subset\p^n_\R$ is called the \emph{absolute quadric}. It may induce a metric in the affine space which then becomes the \emph{Euclidean space} $\E^n_\R$. Note that $\Omega$ consists of imaginary points only.

%Throughout the paper we will work mainly over the field of complex numbers $\C$ and when not stated otherwise we will simply write $\PP^n, \A^n$ and $\E^n$ instead of $\PP^n_{\C}, \A^n_{\C}$ and $\E^n_{\C}$.
%If the dimension of the ambient space is clear from the context, we will write simply $\trans$ instead of $\trans^n$.
%Recall that any projective transformation can be represented as
%\begin{equation}\label{eq proj}
%\varphi:\, \PP^n \rightarrow \PP^n, \, \f x \mapsto C\cdot \f x,
%\end{equation}
%where $C =(c_{ij})_{i,j=0,\ldots,n}$ is a regular matrix.
% If $c_{00}\neq 0$ and $c_{i0}=0$, $i=1,\ldots,n$, then $\eqref{eq proj}$ describes an affine transformation.

Let us denote $\aut(\PP^n_\K)$ the set of all \emph{projective transformations} of $\PP^n_\K$, i.e., $\aut(\PP^n_\K)\cong\PP \GL_{n+1}(\K)$.  A projective transformation mapping the hyperplane at infinity $\omega$ onto itself and thus also $\mathbb{A}^n_\K$ onto itself is called an \emph{affine transformation}. The set of all affine transformations forms a group denoted by $\aff_n(\K)$. Since we will be interested in $\aff_2(\C)$ only we will write simply $\aff$ for this group without any danger of confusion.   A \emph{similarity} is an affine transformation in the Euclidean space $\mathbb{E}^n_\R$ which preserves the absolute quadric. The group of direct similarities is  generally denoted $\simm_n(\R)$. Let us remark that analogously to the affine case we will deal with $\simm:=\simm_3(\R)$ only. Finally,  \emph{isometries} are similarities which preserve distances. 

The matrix representation of any transformation of $\p^n_\K$ can be written in the form
\begin{equation}\label{eq:matrix form}
  \overline{\f A} = \left(\begin{array}{lclll}
    a_{00}&\vline&a_{01}&\cdots&a_{0n}\\
    \hline
    a_{10}&\vline&a_{11}&\cdots&a_{1n}\\
    a_{20}&\vline&\vdots&\ddots&\vdots\\
    a_{n0}&\vline&a_{n1}&\cdots&a_{nn}
  \end{array}\right)
  =
  \left(\begin{array}{lcl}
    a_{00}&\vline&\widehat{\f a}\\
    \hline
    \f a&\vline& \f A
  \end{array}\right).
\end{equation}
The affine transformations correspond to the case $a_{00} \neq 0$ and $\widehat{\f a}=(0,\ldots,0)$.  The additional assumption on the transformation to be a similarity is fulfilled by the condition $\f A^{\top} \f A= \lambda \f I$, where $\lambda \in \R^{>0}$, and especially for $\lambda=1$ we obtain an isometry.

%The affine transformation given by $\eqref{eq proj}$ is an isometry iff the matrix
%$\overline{C}=\left(\frac{c_{ij}}{c_{00}}\right)_{i,j=1,\ldots,n}$ is orthogonal. Moreover if it preserves also the %orientation it is called the \emph{direct isometry} and in this case $\overline{C}$ is a special orthogonal matrix from $%\mathrm{SO}_n$, which is also called the rotation group, because, in dimensions 2 and 3, its elements describes the usual rotations. Given a set $A\in\E_\R^n$, the isometries preserving this set form a group denoted $\sym(A)$, the~so called \emph{group of symmetries}.

Let $\cal G$ be a subgroup of $\aut(\p^n)$ and let $A,B\subset\p^n$. We will write
\begin{equation}
   {\cal G}_{A,B}:=\left\{\phi\in {\cal G}:\  \phi(A)=B\right\}
\end{equation}
for the set of equivalences between $A$ and $B$. In the case $B=A$ the set ${\cal G}_{A,A}$ forms a group and we will denote it ${\cal G}_A$.

\subsection{$\aut(\p^1)$ and its finite subgroups}
As known, any projective transformation $\PP^n\rightarrow\PP^n$ is uniquely determined specifying $n+2$ pairs of points in a general position. In particular, a projective automorphism of $\PP^1$ is specified by three points and their images. Thus for an ordered quadruple  $\{\f a_1,\dots, \f a_4\}$ there exists the~unique $\phi\in\aut(\p^1)$ such that $\phi(\f a_1)=[1:1]$, $\phi(\f a_2)=[0:1]$ and $\phi(\f a_3)=[1:0]$. If we write $\phi(\f a_4)=[s:t]$ then the~\emph{cross-ratio} of the quadruple is defined to be a number
\begin{equation}\label{eq:def cross}
  [\f a_1,\f a_2;\f a_3, \f a_4] =\frac{t}{s}.
\end{equation}

\begin{proposition}\label{prp:cross}
Two ordered quadruples $\{\f a_1,\dots, \f a_4\}$ and $\{\f b_1,\dots, \f b_4\}$ in $\PP^1$ are projectively equivalent if and only if $[\f a_1,\f a_2;\f a_3, \f a_4]=[\f b_1,\f b_2;\f b_3, \f b_4]$.
\end{proposition}

The {\em Riemann sphere} $\widehat{\C}$ is the set $\C \cup \{ \infty \}$, where $\infty$ is a formal point not in $\C$. The homeomorphism $\varphi: \widehat{\C} \rightarrow \PP^1$ given by $\varphi(z) = [1:z]$ for $z \in \C$ and $\varphi(\infty) = [0:1]$ identifies the Riemann sphere $\widehat{\C}$ with the {\em complex projective line} $\PP^1$. Moreover since the {\em stereographic projection} naturally identifies the Riemann sphere  with the unit sphere ${\cal S}^2\subset\E^3_\R$, we obtain three homeomorpic copies of the sphere: the sphere ${\cal S}^2$ itself, the Riemann sphere $\widehat{\C}$ and the projective line $\PP^1_{\C}$. In this way the group $\aut(\p^1)$ can be identified with conformal homeomorphisms of the sphere, see \cite[pg. 26]{sh97} for detailed explanation and for the proof of the following proposition.

%The real plane $\R^2$ can be indentified  with the field of complex numbers $\C$ via $x + i y \leftrightarrow (x,y)$. The {\em Riemann sphere} $\hat{\C}$ is the set $\C \cup \{ \infty \}$, where $\infty$ is formal point not in $\C$. Since the {\em stereographic projection} naturally identifies ${\cal S}^2 \setminus \{ \f n\}$ ($\f n$ is the north pole of the sphere) with $\R^2$, the Riemann sphere can be indentified with ${\cal S}^2$. Moreover the homeomorphism $\varphi: \hat{\C} \rightarrow \PP^1(\C)$ given by $\varphi(z) = [z:1]$ for $z \in \C$ and $\varphi(\infty) = [1:0]$ identifies the {\em complex projective line} $\PP^1(\C)$ with the Riemann sphere $\hat{\C}$. Summing up we have three homeomorpic copies of the sphere: The sphere ${\cal S}^2$ itself, the Riemann sphere $\hat{\C}$ and the projective line $\PP^1(\C)$. And each of them has its advantages.

%The projective transformations of complex projective line $\PP^1_{\C}$ are known as \emph{M\"{o}bius transformations}. A general M\"{o}bius transformation has the form
%\begin{equation}
%f(z) = \frac{a z + b}{c z + d}, \quad a,b,c,d \in \C \,\, \mathrm{and} \,\, a d - b c \neq 0.
%\end{equation}
%Hence these transformations describe automorphisms of sphere $\mathcal{S}^2$, which are exactly the rotations.

\begin{proposition}
Any finite automorphism group of the sphere is conjugate to the rotation group.
\end{proposition}

Moreover all the types of finite rotation groups are classified as follows

\begin{proposition}\label{prp:platonic}
Each finite rotation group of the sphere is isomorphic to one of the following groups:
\begin{enumerate}
\item cyclic groups $\mathcal{C}_n$,
\item dihedral groups $\mathcal{D}_n$,
\item the symmetry  groups of tetrahedron $\mathcal{T}$, octahedron $\mathcal{O}$ or icosahedron $\mathcal{I}$.
\end{enumerate}
\end{proposition}

%A \emph{cyclic group} is a group that is generated by a single element, i.e.,
%\begin{equation}
%\mathcal{C}_n = \langle s : s^n = 1 \rangle \quad  \mathrm{for} \quad n \geq 1.
%\end{equation}
%A \emph{dihedral group} is the group of symmetries of a regular $n$-gon, which includes rotations and reflections
%\begin{equation}
%\mathcal{D}_n = \langle s,t : s^n = t^2 = 1, tst = s^{-1} \rangle\quad  \mathrm{for} \quad n \geq 2.
%\end{equation}
%The symmetry group of the tetrahedron is isomorphic to the \emph{alternating group} $\mathcal{A}_4$, i.e., group of even permutations of a four-elements set. The symmetry group of the cube and octahedron are isomorphic to the \emph{permutation group} $\mathcal{S}_4$ on four elements and finally, the symmetry groups of the dodecahedron and icosahedron are isomorphic to the alternating group $\mathcal{A}_5$.

\subsection{Grassmannians}

The set of all projective subspaces of dimension $k$ in $\p^n_\K$ forms a projective variety; the so called \emph{Grassmannian} $\G(k,n)$. We are mainly interested in two cases. First, the Grassmannian $\G(n-1,n)$ of all subspaces of dimension $n-1$ in $\p^n$, which is again a projective space, called the \emph{dual space} and denoted $\left(\p^n_\K\right)^\vee$.  Second, the variety of lines in $\p^3_\K$, which is a quadratic hypersurface in $\p^5$ and denoted  $\G=\G(1,3)$. For an  introduction to the theory of Grassmannians see  e.g. \cite{PoWa01} or \cite{Ha92}.

Let us focus on the group of automorphisms of $\G$. Any projective transformation $\phi:\PP^3\rightarrow\PP^3$ maps lines to lines. It turns out that it induces a transformation of $\PP^5$ preserving the Grassmannian~$\G$. In fact it induces an injective group homomorphism $\aut(\p^3_\K)\rightarrow\aut(P^5_\K)_\G$, we will write $\widehat{\phi}$ for the transformation induced by $\phi$. Let $\aut(\p^5_\K)^+_\G$ denotes the image of $\aut(\p^3_\K)$ under this homomorphism. Then it is a subgroup of index 2 in $\aut(\p^5_\K)_\G$. Its complement $\aut(\p^5_\K)^-_\G$ is formed by transformations induced by regular projective mappings $\PP^3\rightarrow(\PP^3)^\vee$, cf. \cite[Theorem 10.19]{Ha92}.

\section{Projective equivalences of finite subsets of $\PP^1$}

The paper is devoted to studying selected situations when the computation of projective (or other) equivalences of certain algebraic curves and surfaces  can be simply solved with a unifying approach for determining projective equivalences of associated finite sets on the projective line. Hence, in this section we formulate two algorithms devoted to the detection of equivalences of finite point sets on the projective line given either directly or as the roots of a polynomial.

\subsection{Finite subsets as the collections of points}

Consider two finite subsets $A$ and $B$ of $\PP^1$, obviously they can be projectively equivalent  only if they have the same cardinality.  Since any transformation of the~projective line is determined by three points we see that $\aut(\p^1)_{A,B}$ is non-empty whenever $\#A=\#B=3$. In fact in this case it is isomorphic to the~permutation group on three elements.  We have already seen (Proposition~\ref{prp:cross}) that four point sets are not projectively equivalent in general -- surprisingly if there exists a projectivity mapping $A$ to $B$ then it is not unique.

\begin{lemma}\label{lem:num of aut}
  Let $A$ and $B$ be two subsets of $\PP^n$ such that $\aut(\p^n)_{A,B}$ is non-empty and finite. Then $\#\aut(\p^n)_{A,B}=\#\aut(\p^n)_{A}=\#\aut(\p^n)_B$
\end{lemma}

\begin{proof}
 To prove $\#\aut(\p^n)_{A,B}=\#\aut(\p^n)_{A}$ we construct the mapping $\aut(\p^n)_{A}\rightarrow\aut(\p^n)_{A,B}$ as follows. Fix $\phi\in\aut(\p^n)_{A,B}$ and define $\psi\mapsto\phi\circ\psi$. It is easily seen that it is bijective. The second part is analogous.
\end{proof}

\begin{lemma}
  Let $A=\{\f a_1,\dots,\f a_4\}\subset\PP^1$ be a set consisting of four distinct points. Then
  $\aut(\p^1)_A$ is a group of order at least four.  More precisely $\aut(\p^1)_A=\Z_2\times \Z_2$ unless $[\f a_1,\f a_2;\f a_3,\f a_4]\in\left\{-1,2,\frac{1}{2},\e^{\pm\I\pi}\right\}$
\end{lemma}

\begin{proof}
Altogether there are 24 permutations on four elements, whereas there exist at most six different values of their cross-ratios. Thus by Proposition~\ref{prp:cross} the group $\aut(\p^1)_A$ has the order at least four. Moreover when the cross-ratio is different from $-1,2,\frac{1}{2},\e^{\pm\I\pi}$ then there are exactly six values and it is known that the subgroup of permutations on four points preserving their cross ratio is the Klein group $\Z_2\times\Z_2$.
\end{proof}

On contrary if two projectively equivalent sets have more than four points then the transformation is generically unique.

\begin{algorithm}[H]
\caption{Detection of equivalences of finite sets using cross-ratios.}\label{alg:finite cross} \algsetup{indent=2em}
\begin{algorithmic}[1]
 \REQUIRE  $A=\{\f a_1,\dots,\f a_k\}\subset\PP^1$, $B=\{\f b_1,\dots,\f b_k\}\subset\PP^1$

 \STATE
 Compute the cross-ratio $\lambda=[a_{i_1},a_{i_2};a_{i_3},a_{i_4}]$ of an arbitrary ordered quadruple from $A$.

 \STATE
 Denote $\Lambda=\left\{\lambda,\frac{1}{\lambda},1-\lambda,\frac{1}{1-\lambda},\frac{\lambda}{\lambda-1},\frac{\lambda-1}{\lambda}\right\}$.

 \STATE
 Compute $k\choose 4$ cross-ratios $\lambda_j=[b_{j_1},b_{j_2};b_{i_3},b_{j_4}]$ for all non-ordered quadruples from $B$.

 \STATE
 Denote $Q_B$ the set of all quadruples $\{b_{j_1},b_{j_2};b_{j_3},b_{j_4} \}$ for which $\lambda_j\in\Lambda$.

  \IF{$Q_B = \emptyset$}
   \STATE
     $\aut(\p^1)_{A,B}=\emptyset$
   \ELSE
   \STATE
    For each element $\{b_{i_1},b_{i_2},b_{i_3},b_{i_4}\}\in Q_B$ find the corresponding  automorphism $\phi_{\ell}$ mapping $\{a_{i_1},a_{i_2},a_{i_3},a_{i_4}\}$ to $\{b_{i_1},b_{i_2},b_{i_3},b_{i_4}\}$.
   \STATE
    When $\phi_{\ell}$ maps all points of $A$ onto all points of $B$ include such an automorphism  into $\aut(\p^1)_{A,B}$.

   \ENDIF

 \ENSURE $\aut(\p^1)_{A,B}$.
\end{algorithmic}
\end{algorithm}

\begin{example}\rm
Consider two sets of five points
\begin{equation}
 A = \left\{
[-1+2 \I : 2+\I ],
 [4 : -4-6 \I],
 [-5-4 \I : 7 \I],
 [-6-4 \I : 8 \I],
 [3+\I : -1-5 \I]
 \right\}
\end{equation}
and
\begin{equation}
 B = \left\{
 [5 \I : 1-2 \I ],
 [1-5 \I : 1],
 [-7-2 \I : 4 \I],
 [-9-\I : 1+4 \I],
 [3-3 \I : 0]
 \right\}
\end{equation}
Computing all the cross-ratios of the non-ordered quadruple composed of the first four elements of $A$ yields
\begin{equation}
\left\{\frac{1}{10}-\frac{\I}{5},2+4 \I,\frac{9}{10}+\frac{\I}{5},\frac{18}{17}-\frac{4 \I}{17},-\frac{1}{17}+\frac{4 \I}{17},-1-4 \I\right\}.
\end{equation}
Now, we compute the cross-ratios for all subsets of $B$ with four elements
\begin{equation}
\left\{\frac{1}{10}-\frac{\I}{5},-4 \I,-1-4 \I,1-\frac{\I}{4},\frac{19}{20}-\frac{\I}{40}\right\}.
\end{equation}
Since the cross-ratios $\frac{1}{10}-\frac{\I}{5}$ and $-1-4 \I$ are contained in both sets we have altogether $8$ (four for each cross-product) candidates determining the transformation. However only one couple of quadruples, in particular
\begin{equation}
\left\{
 [-1+2 \I : 2+\I ],
 [4 : -4-6 \I],
 [-5-4 \I : 7 \I],
 [-6-4 \I : 8 \I]
 \right\}
\end{equation}
and
\begin{equation}
\left\{
 [5 \I : 1-2 \I ],
 [1-5 \I : 1],
 [-7-2 \I : 4 \I],
 [-9-\I : 1+4 \I]
 \right\}
\end{equation}
determines a projective transformation represented by the matrix
\begin{equation}
\overline{\f A} =
\left(
\begin{array}{cc}
 1-5 \I & -2+3 \I \\
 -1-3 \I & 1+2 \I \\
\end{array}
\right).
\end{equation}
Since this transformation maps the remaining point from $A$ to the remaining point from $B$, it is the correct one.
\end{example}

\subsection{Finite sets as the~roots of a~polynomial}\label{sub:conjugate points}

A collection of $n$ points in $\PP^1$ can be given as a set of roots of homogeneous form of degree $n$. However, the forms carry more information because of possible higher multiplicities of its roots. Let $F(x_0,x_1)$ and $G(x_0,x_1)$ be the forms. Any $\phi\in\aut(\p^1)$ acts naturally on the set of forms of degree $n$ by $F\mapsto F\circ\phi$.  We define
\begin{equation}\label{eq:equiv_forms}
  \aut(\p^1)_{F,G}=\{\phi\in\aut(\p^1)\mid  \exists\lambda\in\C^*:\ G\circ\phi=\lambda F\}.
\end{equation}
If $A$ and $B$  are the sets of roots of $F(x_0,x_1)$ and $G(x_0,x_1)$ respectively, then obviously any $\phi\in\aut(\p^1)_{F,G}$ induces a transformation mapping $A$ to $B$. Thus $\aut(\p^1)_{F,G}\subset\aut(\p^1)_{A,B}$. If all the roots of $F$ and $G$ are simple then we have the equality. Nonetheless the inclusion may be proper,  in general .

\begin{algorithm}[H]
\caption{Detection of equivalences of finite sets given by polynomial relation.}\label{alg:finite poly} \algsetup{indent=2em}
\begin{algorithmic}[1]
 \REQUIRE $F(x_0,x_1)$ and $G(x_0,x_1)$

  \STATE
  Consider the ideal $I$ generated by the coefficients of the polynomial $G(a_{00} x_0 + a_{10} x_1, a_{01} x_0 + a_{11} x_1)- F(x_0,x_1)$ w.r.t. $x_0,x_1$.

  \STATE
  Compute the reduced Gr\"obner basis  $GB=\{g_1,\ldots,g_{\ell}\}$ of $I$ w.r.t. a suitable ordering of the variables $a_{00}, a_{10}, a_{01}, a_{11}$.

  \IF{$GB=\{ 1 \}$}
   \STATE
     $\aut(\p^1)_{F,G}=\emptyset$
   \ELSE
   \STATE
    Find a solution of the system of equations $g_1=0,\ldots,g_{\ell}=0$.
   \STATE
    $\aut(\p^1)_{F,G}$ consists of automorphisms given by all the solutions $a_{00}, a_{10}, a_{01}, a_{11}$
   \ENDIF

\ENSURE $\aut(\p^1)_{F,G}$.

\end{algorithmic}
\end{algorithm}

\begin{remark}\rm
In Algorithm \ref{alg:finite poly} we compute the Gr\"obner basis of the ideal generated by the coefficients of $G(a_{00} x_0 + a_{10} x_1, a_{01} x_0 + a_{11} x_1)- F(x_0,x_1)$, whereas two forms are projectively equivalent if one can be mapped to the other up to a complex multiple $\lambda$, cf. \eqref{eq:equiv_forms}. However we do not need to consider the additional parameter $\lambda$ (which would cost some computational time) since the matrix $\f A$ of the transformation is also determined uniquely up to a complex multiplication and hence it can ensure $\lambda = 1$.
\end{remark}

Although the Algorithm~\ref{alg:finite poly} requires to solve a~large system of non-linear equations, recall that two general forms of degree at least four are not projectively equivalent and thus we have the ideal $I=\langle 1\rangle$. For example, in this case CAS Mathematica or CAS Maple give a decision even for polynomials of degree 20 within a few seconds. Next,  the general polynomial of degree at least five possesses no projective automorphism and thus by Lemma~\ref{lem:num of aut} the transformation is unique for two equivalent generic forms of high degree. Again in this case the answer is obtained within a few seconds as the Gr\"obner basis possesses a special structure containing linear forms in $a_{00}, a_{10}, a_{01}, a_{11}$ which uniquely determine the automorphism. In addition, the basis also contains one nonlinear term responsible for a particular choice of $a_{00}, a_{10}, a_{01}, a_{11}$ (of course, describing the same automorphism for all choices), cf. Example~\ref{ex:2forms} and \eqref{eq:GB}.

\begin{example}\rm\label{ex:2forms}
Consider two forms of degree six
\begin{equation}
F=571 x_0^6-426 x_1 x_0^5-1827 x_1^2 x_0^4+8532 x_1^3 x_0^3-11259 x_1^4 x_0^2+12150 x_1^5 x_0-3645 x_1^6
\end{equation}
and
\begin{equation}
G=-569 x_0^6+430 x_1 x_0^5+1758 x_1^2 x_0^4+3891 x_1^3 x_0^3+6054 x_1^4 x_0^2+2105 x_1^5 x_0+2055 x_1^6
\end{equation}
The Gr\"obner basis of the ideal generated by the coefficients of
\begin{equation}
G(a_{00} x_0 + a_{01} x_1, a_{10} x_0 + a_{11} x_1)- F(x_0,x_1)
\end{equation}
has the form
\begin{equation}\label{eq:GB}
\left\{1771561 a_{11}^6-46656,a_{10}\gamma -a_{11} ,2 a_{01} +5 a_{11} ,6 a_{00} -7 a_{11} \right\}.
\end{equation}
Since the transformation is unique up to a scalar multiplication, we can omit the first polynomial and solve the linear system only, i.e., we obtain
\begin{equation}
a_{01} \to -\frac{15 a_{00} }{7}, \quad a_{10} \to \frac{6 a_{00}}{7}, \quad a_{11} \to \frac{6 a_{00}}{7}
\end{equation}
yielding the transformation of $\p^1$ represented by the matrix
\begin{equation}
\overline{\f A} = \left(
\begin{array}{cc}
 7 & -15 \\
 6 & 6 \\
\end{array}
\right)
\end{equation}
mapping the roots of $F$ to the roots of $G$.
\end{example}

\section{Projective and other equivalences of selected algebraic varieties}

In this section we will discuss several problems which can be reduced to the computation of equivalences of finite sets in $\p^1$. The general setting of our problem is following. Let be given ${\cal G}$ a subgroup of $\aut(\p^n)$ and two algebraic varieties $X,Y\subset\p^n$. Our goal is to compute ${\cal G}_{X,Y}$ using the approach introduced in the previous section. Hence we find suitable forms $F(x_0,x_1)$ and $G(x_0,x_1)$ associated to the varieties $X$ and $Y$ together with the inclusion
\begin{equation}
   \iota:{\cal G}_{X,Y}\hookrightarrow \aut(\p^1)_{F,G}.
\end{equation}
The idea is that for $\phi\in \aut(\p^1)_{F,G}$ it is simple to decide whether there exists $\psi\in{\cal G}_{X,Y}$ such that $\phi=\iota(\psi)$.  In this sense $\aut(\p^1)_{F,G}$ consists of candidates for possible mappings $X\rightarrow Y$. If $\aut(\p^1)_{F,G}$ is not too large compared to ${\cal G}_{X,Y}$ -- in particular if they have the same dimension, then we will call it a~\emph{good candidate set} of ${\cal G}_{X,Y}$.

\subsection{Projective equivalences of rational curves}\label{sec:curves}

When studying projective equivalences of  algebraic varieties then it is natural to start with the further simplest case after the collections of points, i.e., with rational curves. Recently, \cite{HaJu18} published a paper devoted to the detection of equivalences and symmetries of rational curves with respect to the group of projective transformations including the subgroup of affine transformations. We continue in this investigation and present an algorithm based on computing projective equivalences of finite point sets.

By a rational curve of degree $d$ in $\p^n$ we mean the image of a~morphism $\PP^1\rightarrow\PP^n$ given by
\begin{equation}\label{eq:param}
  \f p(s,t) = \left[p_0(s,t):p_1(s,t):\cdots:p_n(s,t)\right],
\end{equation}
where $p_i(s,t)$ are homogeneous polynomials of degree $d$ without a common factor.  Moreover the mapping is assumed to be a~birational morphism, i.e,  it is almost everywhere injective. In what follows, we assume the~curve to be \emph{non-degenerate}, i.e., it is not contained in any hyperplane, or equivalently  all the polynomials $p_i$ are linearly independent over $\C$. Obviously a curve can be non-degenerate only if $d\geq n$.

Since the degree of a~curve is a projective invariant, the equivalent curves must have the same degree. Recall that any parameterization $\f p(s,t)$ of a rational curve of degree $d$ in $\PP^n$ is an~image of the~rational normal curve
\begin{equation}
  C_d:\quad \f c_d(s,t)=\left[s^d:s^{d-1} t:\cdots t^d\right]
\end{equation}
under some projection $\PP^d\backslash M\rightarrow\PP^n$, where $M$ is a linear subspace of dimension $d-n-1$. In particular, if $p_i(s,t)= \sum_{j=0}^d p_{ij} s^{d-j}t^j$, then the projection is given by the matrix $(p_{ij})_{i,j=0}^{n,d}$ and the subspace $M$ is generated by the kernel of the matrix. Clearly there exists a projective transformation taking one parameterization to the other one if and only if the projection matrices have the same kernels.

\begin{proposition}
Two parameterizations of non-degenerate rational curves of degree $n$ in $\PP^n$ are always projectively equivalent.
\end{proposition}

%In particular, two parametrizations of non-degenerate rational curves of degree $n$ in $\PP^n$ are always projectively equivalent. However the parameterization of the curve is unique only up to a~reparameterization. And thus we obtain:
%
%
%\begin{lemma}
%  Let $C,D\subset\PP^n$ be two rational non-degenerate curves of degree $n$. Then $\aut(\p^n)_{C,D}\cong\aut(\p^1)$.
%\end{lemma}
%
%\begin{proof}
%
%
%  \todo[inline]{define homomorphism $\PP GL(2,\C)\rightarrow\PP GL(n+1,\C)$ and find its kernel}
%\end{proof}

Hence in what follows we focus on transformations between curves of degree $d>n$. Let $C\subset\PP^n$ be parameterized by $\f p(s,t)$ then the \emph{osculating $k$-planes}, $k=1,\ldots,n-1$, having the contact of order at least $k+1$ are spanned by $\frac{\partial^k \f p(s,t)}{\partial s^k},\frac{\partial^k \f p(s,t)}{\partial s^{k-1}\partial t},\dots,\frac{\partial^k \f p(s,t)}{\partial t^k}$. For $k=n-1$ we obtain \emph{osculating hyperplanes}. \emph{Stall points} are the points where the osculating hyperplane has the~contact higher than expected. They are given by the condition
\begin{equation}\label{oscul}
\Delta_\f p(s,t) =  \det\left[\frac{\partial^n \f p(s,t)}{\partial s^n},\frac{\partial^n \f p(s,t)}{\partial s^{n-1}\partial t},\cdots,\frac{\partial^n \f p(s,t)}{\partial t^n}\right]=0.
\end{equation}
In particular, the homogeneous form $\Delta_\f p(s,t)$ has degree $(d-n)(n+1)$ and thus on any non-degenerate curve with degree $d>n$, there exist only finitely many stalls. A projective transformation takes osculating $k$-planes of the curve to osculating $k$-planes of its image curve, in particular stalls are mapped to stalls.

\begin{theorem}\label{thm:cand curves}
  Let $C,D\subset\p^n$ be non-degenerate rational curves of degrees $d>n$ and let $\f p:\p^1\rightarrow C$ and $\f q:\p^1\rightarrow D$ be the birational morphisms parameterizing them. Then $\aut(\p^1)_{\Delta_\f p,\Delta_\f q}$ is a candidate group for $\aut(\p^n)_{C,D}$ and the inclusion $\iota:\aut(\p^n)_{C,D}\hookrightarrow\aut(\p^1)_{\Delta_\f p,\Delta_\f q}$ is given by $\iota:\phi\mapsto \f q^{-1}\circ\phi\circ\f p$.
\end{theorem}

\begin{proof}
Assume that there exists a projective transformation $\phi$ taking the curve $C: \f p(s,t)$ to the curve $D: \f q(u,v)$ and thus there exists the reparameterization $\psi$ making the following diagram commutative
\begin{equation}\label{eq:diag reparam}
\begin{array}{c}
\xymatrix{
 C  \ar[r]^{\phi} & D\\
 \PP^1 \ar[r]^{\psi}\ar[u]^{\f p} & \PP^1\ar[u]_{\f q}
}
\end{array}.
\end{equation}
\end{proof}

Now, given $\psi\in\aut(\p^1)_{\Delta_\f p,\Delta_\f q}$ we would like to decide whether it is the image of some $\phi\in\aut(\p^n)_{C,D}$.  By \eqref{eq:diag reparam} this happens if and only if the parameterizations $\f p$ and $\f q\circ\psi$ are projectively equivalent.  Nonetheless this is equivalent to the condition that the matrices of the coefficients of these parameterizations have the same kernels. The method is summarized in Algorithm~\ref{alg:rat curves}.

\begin{algorithm}[H]
\caption{Projective equivalences of rational curves.}\label{alg:rat curves} \algsetup{indent=2em}
\begin{algorithmic}[1]
\REQUIRE Curves $C: \f p(s,t)$ and $D: \f q(s,t)$

\STATE Compute the forms $\Delta_\f p(s,t)$ and $\Delta_\f q(s,t)$, cf.~\eqref{oscul}.
\STATE Find the candidate group $\aut(\p^1)_{\Delta_\f p,\Delta_\f q}$ composed of automorphisms $\psi$ described by matrices $\overline{\f B}_{\psi}$, see Section~\ref{sub:conjugate points}.
\STATE For all $\psi \in \aut(\p^1)_{\Delta_\f p,\Delta_\f q}$, compute the kernels $K_{\f p}$ and $K_{\f q\circ\psi}$ of the matrices of the coefficients of $\f p(s,t)$ and $\f q(\overline{\f B}_{\psi} (s,t)^{\top})$.
\IF{$K_{\f p} = K_{\f q\circ\psi}$}
   \STATE
    By solving linear equations corresponding to $\overline{\f A} \f p(s,t) - \f q(\overline{\f B}_{\psi} (s,t)^{\top})$  compute the projective transformation $\phi$ given by $\overline{\f A}$ and include it into $\aut(\p^n)_{C,D}$.
   \ENDIF

 \ENSURE $\aut(\p^n)_{C,D}$.
\end{algorithmic}
\end{algorithm}

The case of projective equivalences between rational quartics in $\p^4$  was already studied in \cite[pg. 42]{Te36}, with the result that two quartics with same osculating polynomial (up to reparameterization) are projectively equivalent.  Let us briefly recall the arguments. For a quartic parameterization $\f p(s)=[p_0(s):\cdots :p_3(s)]$ of $C$ the condition that $\f p(s_i)$, $i=1,\dots,4$ are coplanar is symmetric algebraic relation in $s_i$ which is moreover linear in each parameter -- because when given three points on $C$ then the fourth point is determined uniquely. Hence this relation has the form
\begin{equation}\label{eq:magic}
  \Phi(s_1,s_2,s_3,s_4)=\varphi_4s_1s_2s_3s_4+\varphi_3\sum_{i<j<k} s_is_js_k+\varphi_2 \sum_{i<j}s_is_j+\varphi_1\sum_i s_i+\varphi_0.
\end{equation}
This polynomial is unique up to a scalar and it is a polarized form of a quartic polynomial
\begin{equation}
  \Phi(s,s,s,s)=\varphi_4s^4+4\varphi_3s^3+6\varphi_2s^2+4\varphi_1s+\varphi_0.
\end{equation}
The roots of this polynomial correspond to points on the curve $C$ with their osculating plane of contact order four. Thus this is precisely (up to a scalar multiple) our osculating polynomial $\Delta_\f p(s)$.

From the construction of \eqref{eq:magic} it is clear that it is invariant under projective transformations. Conversely let be given two curves with the same osculating polynomial and thus with the same \eqref{eq:magic} as well. Then there exists a correspondence between quadruples of coplanar points and thus between planes.  This provides a projective transformation taking one curve to the other. This proof, in fact, does not work only for quartics in $\p^3$ but for any non-degenerate rational curve of degree $n+1$ in $\p^n$, So we arrive at the following proposition:

\begin{proposition}\label{Prop:Telling}[Telling]
  Let $C,D\subset\p^n$ be rational curves of degree $n+1$ and let $\aut(\p^1)_{\Delta_\f p,\Delta_\f q}$   be a candidate group. Then the mapping $\iota:\aut(\p^n)_{C,D}\hookrightarrow\aut(\p^1)_{\Delta_\f p,\Delta_\f q}$ is a~bijection.
\end{proposition}

\begin{example}\rm
Consider two rational quartics in $\p^3$
\begin{multline}
C: \f p(s,t) = \left[75 s^4-296 s^3 t+424 s^2 t^2-272 s t^3+64 t^4:9 s^4-16 s^3 t-8 s^2 t^2+32 s t^3-16 t^4:
\right. \\ \left.
13 s^4-20 s^3 t-8 s^2 t^2+32 s t^3-16 t^4 : -53
   s^4+104 s^3 t-40 s^2 t^2-48 s t^3+32 t^4\right]
\end{multline}
and
\begin{multline}
D: \f q(s,t) = \left[32 s^4+96 s^3 t+64 s^2 t^2+36 s t^3+9 t^4:-80 s^4-128 s^3 t-48 s^2 t^2-4 s t^3+7 t^4:
\right. \\ \left.
-32 s^4-32 s^3 t+16 s^2 t^2+16 s t^3+6 t^4:64
   s^4+160 s^3 t+144 s^2 t^2+64 s t^3+10 t^4\right].
\end{multline}
First, we compute the osculating polynomials
\begin{equation}
\Delta_\f p = 3 s^4+20 s^3 t-72 s^2 t^2+64 s t^3-16 t^4
\end{equation}
and
\begin{equation}
\Delta_\f q = 8 s^4+24 s^3 t+12 s^2 t^2-2 s t^3-t^4.
\end{equation}
Now, employing Algorithm~\ref{alg:finite poly} we obtain four different reparameterizations of  $\Delta_\f q$ yielding $\Delta_\f p$. Using Proposition~\ref{Prop:Telling}, we know that to each reparameterization there will exist a corresponding projective transformation mapping $C$ to $D$. For the sake of brevity we present only one case, e.g., the reparameterization
%\begin{equation}
%(s,t) \to \overline{\f B} (s,t)^{\top}, \quad \overline{\f B} = \left(
%\begin{array}{cc}
% 1 & 4  \\
% -2 & -4 \\
%\end{array}
%\right).
%\end{equation}
\begin{equation}
s\mapsto s-2 t, \quad t\mapsto 4 t-4 s
\end{equation}
leads to the following projective transformation given by the matrix
\begin{equation}
\overline{\f A} =\left(
\begin{array}{cccc}
 1 & 13 & 16 & 2 \\
 9 & -16 & -24 & -10 \\
 4 & -18 & -6 & -6 \\
 -2 & -12 & 10 & -4 \\
\end{array}
\right).
\end{equation}
\end{example}

\subsection{Projective equivalences of rational ruled surfaces}
Let $S$ be a rational ruled surface in $\p^3$, i.e., a surface generated by a rational one-dimensional family of lines. Such a family is parameterized by a rational curve on the Grassmannian $\G\subset\p^5$.
Hence it is tempting to use the methods from the previous section to study the ruled surfaces as well. Formally, a \emph{rational ruled surface} $\av{S}$ is a projection of the~\emph{rational normal scroll} $\Sigma_{d_1,d_2}$ to $\PP^3$ for some $0\leq d_1\leq d_2$, where the rational normal scroll is the surface in $\PP^{d_1+d_2+1}$ parametrized as
\begin{equation}
  \left[ 1:s:\cdots s^{d_1}:t:t s:\cdots:ts^{d_2} \right].
\end{equation}
The degree of $\av{S}$ equals $d=d_1+d_2$.  Since the projective automorphisms of the plane and quadrics are well known, we will focus on the case $d>2$ only. Then the surface contains exactly one  one-dimensional family of lines, the so called \emph{rulings}. The image  of this family on the Grassmannian $\G$ is then a rational curve of degree $d$, and we will denote it by $\widehat{\av{S}}$.

Recall that each projective transformation $\phi:\p^3\rightarrow\p^3$ induces a transformation $\widehat{\phi}:\p^5\rightarrow\p^5$ preserving the Grassmannian $\G$ and that the map $\pi:\aut(\p^3)\rightarrow\aut^+(\G)$ is an isomorphism.

\begin{lemma}
  For two rational ruled surfaces $R$ and $S$, the  restriction map $\pi|_{\aut(\p^3)_{S,R}}: \aut(\p^3)_{S,R}\rightarrow\aut^+(\G)_{\widehat{S},\widehat{R}}$ is  bijective.
\end{lemma}

\begin{proof}
  It is clear that the map indeed takes the set $\aut(\p^3)_{S,R}$ onto $\aut^+(\G)_{\widehat{S},\widehat{R}}$.  The bijectivity then follows from the fact that it is a restriction of the isomorphism.
\end{proof}

We saw that two-nondegenrate rational cubics in $\p^3$ were always projectively equivalent, whereas this was not true for quartics any more. Let us investigate the same question for rational ruled surfaces in $\p^3$, too. There are two possibilities: $S$ is either a projection of $\Sigma_{0,3}$ or $\Sigma_{1,2}$, where the first one is a cone  over a rational cubic. The theory of projective equivalences between cones is clearly equivalent to the theory of projective equivalences of planar curves. Two rational planar cubics are projectively equivalent  whenever they have equivalent their osculating polynomials, by Proposition~\ref{Prop:Telling}. Since a general rational cubic has its osculating polynomial with three distinct roots (the case of a~nodal cubic) we conclude that two generic projections of $\Sigma_{0,3}$ are projectively equivalent.

The generic projection of $\Sigma_{1,2}$ contains a chain of four special lines, see \cite{Pi05} for numeric formulas for the degree of singular locus, number of torsal lines, etc. The singular locus of $S$ is a line $\Gamma$. Through each point of $\Gamma$ there pass two rulings, except of two pinch points $\f g_i\in\Gamma$, $i=1,2$, where there is only one ruling $L_i$ counted with the~multiplicity two (the so called \emph{torsal ruling}). Except of the singular line there exists  another unique  line $\av{M}$, not belonging to the family of rulings on the surface $S$. Namely it is the projection of the~line $(1:s:0:0:0)$ on $\Sigma_{1,2}$. Let $\f m_i$ denote the intersection points $\av{M}\cap\av{L}_i$.  The lines $\Gamma$ and $M$ are skew.

Since a projective transformation of $\p^3$ is given by $5$ points there exists a transformation such that
\begin{equation}
\f g_1=[1:0:0:0],\ \f g_2=[0:1:0:0],\ \f m_1=[0:0:1:0]\ \text{and}\ \f m_2=[0:0:0:1].
\end{equation}
The surface is then obtained by joining corresponding points on lines $M$ and $
\Gamma$. W.l.o.g. parameterize  $M$ homogeneously as $s_1\f m_1+s_2\f m_2$, i.e., the points $\f m_i$ correspond to parameter values $s_j=0$ for $i\not=j$. The double curve is traced twice and thus it admits a parameterization $\f g(s_1,s_2)=g_1(s_1,s_2)\f g_1+g_2(s_1,s_2)\f g_2$ for some quadratic forms $g_i(s_1,s_2)$.  The fact that the points $\f g_i$ lie on torsal rulings means that these points are pinch points on the double line $\Gamma$ and the parameterization $\f g(s_1,s_2)$ fails to be regular at these points. Together with the conditions on compatibility with parameterization of $M$ we arrive at the  possible parameterizations $\alpha s_1^2\f g_1+\beta s_2^2\f g_2$, where $(\alpha:\beta)\in\p^1$.

To sum up we just constructed a family of ruled surfaces admitting a parameterization (written non-homogeneously)
\begin{equation}
   [s:1:\alpha s^2 t:\beta t];
\end{equation}
Nevertheless it is easy to see that all such surfaces are projectively equivalent. Let us just mention that this analysis shows that the group of automorphisms of a generic projection of $\Sigma_{1,2}$ is two-dimensional and has two disconnected components. The  first one is formed by transformations preserving all the lines $M$, $\Gamma$ and $L_i$, whereas the second component swaps two torsal rulings $L_i$.

Closer look at the group of automorphisms of ruled cubic reveals that there exist non-trivial transformations which preserve each ruling. To imagine this, consider at the moment a cylinder in the affine space. It is obviously invariant under translations with the direction of its axis. These transformations do not interchange the rulings of the surface.  In other words if $\widehat{\phi}\in\aut^+(\G)$ is the associated transformation and $\widehat{S}$ the curve on the Grassmannian then $\widehat{\phi}$ is the  identity when restricted to $\widehat{S}$. For a rational ruled surface we define
\begin{equation}
  \mathcal{N}_\av{S}=\left\{\phi\in\aut(\p^3)_S:\ \widehat{\phi}|_{\widehat{S}}=\mathrm{id}\right\}.
\end{equation}
The reason for the existence of non-trivial ${\cal N}_S$ is the fact that the curve $\widehat{S}$ can be contained in the subspace of dimension less than $5$.  Let us write $\mathrm{span}(\widehat{S})$ for the smallest subspace containing $\widehat{S}$. If $\phi\in\mathcal{N}_S$ then the projective transformation $\widehat{\phi}$ must be the identity on the whole subspace $\mathrm{span}(\widehat{S})$.  We will briefly discuss properties of the group $\mathcal{N}_S$ in dependence on the dimension of this subspace.

\paragraph{$\dim\mathrm{span}(\widehat{S})=5$}
In this case any $\widehat{\phi}$ which is identity on $\widehat{S}$ must be the~identity on the whole space $\p^5$.  Therefore $\mathcal{N}_S=\{\mathrm{id}\}$.

\paragraph{$\dim\mathrm{span}(\widehat{S})=4$}
Now, $\mathrm{span}(\widehat{S})$ is a hyperplane in $\p^5$. The section of the $\G$ by a hyperplane is called a~\emph{linear complex}. For an introduction to the theory of linear complexes see e.g. \cite[Chapter 3]{PoWa01}. Here we will recall only some necessary notions. There exist two different kinds of complexes. Since $\G$ is a hyper-quadric in $\p^5$ it induces a correspondence between points and hyperplanes. It associates to each point its polar hyperplane and vice versa. If the point $\f p$ is contained in $\G$ then the section by the polar hyperplane is called \emph{singular complex}. Otherwise the complex is said to be \emph{regular}.

Let $H\subset\p^5$ be a hyperplane. A transformations $\mu\in\aut(\p^5)$ leaving all points of $H$ invariant is called \emph{perspective collineation} and there exists a point $\f p\in\p^5$ such that for each $\f x\in\p^5$ the triple $\f p$, $\f x$ and $\mu(\f x)$ is collinear -- see \cite[Theorem 1.1.9]{PoWa01}. The point $\f p$ is called a center. The additional requirement that the quadric $\G$ must be invariant under $\mu$ as well, forces $\f p$ to be the pole of the hyperplane w.r.t. $\G$.

Thus for a singular complex the only transformation preserving $\G$ and leaving each point of the hyperplane invariant is the identity.  There exists additional transformation except of identity in regular case, namely the reflection induced by $\G$ and $H$.  However it is known that such a transformation is induced by a  mapping from $\p^3$ to the dual space $\left(\p^3\right)^\vee$. In other words $\mu\not\in\aut^+(\G)$, see \cite[Section 3.1]{PoWa01} for the detailed discussion.  Hence we conclude that $N_{S}=\{\mathrm{id}\}$ in the case $\dim\mathrm{span}(\widehat{S})=4$, too.

\paragraph{$\dim\mathrm{span}(\widehat{S})=3$} In this case $\mathrm{span}(\widehat{S})\cap\G$ has dimension two and thus it the so called \emph{congruence}. The space polar to $\mathrm{span}(\widehat{S})$ is a line. If the line is contained in $\G$, then it corresponds to a pencil of lines passing through a point in $\p^3$. The polar space then intersects $\G$ in the set of lines intersecting each line in this pencil.  These are the lines which lie in the plane of the pencil or lines passing through the vertex.  Thus  ${\widehat{S}}\cap\G$ is irreducible and consists of union of two 2-planes; each of them corresponding to one type of lines. However $\widehat{S}$ is irreducible and thus it must be contained in exactly one of these planes, which is a contradiction with the assumption $\dim\mathrm{span}(\widehat{S})=3$.

Hence assume that the line polar to $\mathrm{span}(\widehat{S})$ is not contained in the Grassmannian.
Then it intersects it in two (not necessarily) distinct points. These two points (if distinct) correspond to lines $M,N\subset\p^3$ and the congruence in this case consists of their transversals.  Thus a ruled surface is then formed by a one-dimensional family of lines intersecting both $M$ and $N$. A detailed description of the linear of congruences can be found again in \cite[Section 3.2]{PoWa01}.

In this case the group ${\cal N}_S$ is not trivial any more. To see this let $H\subset\p^5$ be a space of dimension $3$ and $H^\perp$ its polar line w.r.t. $\G$. Choose two points in $H^\perp\backslash{\G}$ then a composition of two reflections induced by these points is a transformation preserving $\G$ and fixing every point of $H$. From above we know that the original transformation is a composition $\p^3\rightarrow(\p^3)^\vee\rightarrow\p^3$ and thus it is an element of $\aut(\p^3)$. Moreover since $H$ and $H^\perp$ span the whole space $\p^5$ in this case we can conclude that all these transformations can be naturally identified with projective automorphisms  $H^\perp\rightarrow H^\perp$ leaving the intersection $H^\perp\cap\G$ invariant.  Therefore we conclude that $\dim\mathcal{N}_S=1$ in this case.

\paragraph{$\dim\mathrm{span}(\widehat{S})=2$}
There are basically two options in this case. First, curve $\widehat{S}$ is the~section of $\G$ by the~plane $\mathrm{span}\,(\widehat{S})$, i.e., it is a conic section and thus the surface $S$ is a quadric. Second $\mathrm{span}\,(\widehat{S})$ is contained in $\G$. In this case the surface $S$ is either a plane or a cone. In the conical case the group of all projective transformations of $\p^3$ leaving the vertex invariant can be identified with a group of projective transformations of $\p^2$. Thus $\mathcal{N}_S$ is again non-trivial.

\paragraph{$\dim\mathrm{span}(\widehat{S})=1$}
Since the curve $\widehat{S}$ is a line, the surface $S$ must be a plane in $\PP^3$.

The cases $\deg S=1,2$ were excluded, and thus the only case with non-trivial group occurs for $\dim\mathrm{span}\,(\widehat{S})=3$ or for the cones. Since a generic rational curve of degree at least four is not contained in three-dimensional space, we see that most of surfaces of degree at least four possess the~trivial subgroup $\mathcal{N}_S$.

\subsection{Transformations of affine curves}\label{Sub:affine}

Another problem which can be reduced to the computation of projective equivalences between finite sets of points in $\PP^1$ is the detection of affine transformations mapping a planar curve $C$ to a planar curve $D$, i.e., the computation of $\aff_{C,D}$. Note that for rational curves $C$ and $D$ we could use methods from Subsection~\ref{sec:curves} to detect $\aut(\p^2)_{C,D}$. The affine transformations then form its subset preserving $\omega$.  Hence the curves are not assumed to be necessarily rational in this part. We only require that they are irreducible. Thus they are given by their irreducible defining polynomials $F(x_0,x_1,x_2)$ and $G(x_0,x_1,x_2)$ of degrees $d$.   Write $F(x_0,x_1,x_2)=F_d(x_1,x_2)+F_{d-1}(x_1,x_2)x_0+\cdots F_0(x_1,x_2)x_0^d$ and similarly for the form $G$. As the affine transformations between lines or conics are easy to find we will omit these cases and assume $d>2$.

The affine transformations are exactly the projective transformations $\PP^2\rightarrow\PP^2$ preserving the ideal line $\omega: x_0=0$ and thus the matrix representation of any affine transformation can be written as
\begin{equation}\label{eq:matrix}
\left(
\begin{array}{cccc}
a_{00}&\vline&0&0\\
\hline
a_{10}&\vline& a_{11}&a_{12}\\
a_{20}&\vline& a_{21}&a_{22}\\
\end{array}
\right),\qquad a_{00}\not=0\ \text{and}\ a_{11}a_{22}-a_{12}a_{21}\not=0.
\end{equation}

Write $\f A$ for the matrix $\left(a_{ij}\right)_{ij=1}^2$. The affine transformation acts on $\omega$ via $[x_1:x_2]^{\top}\mapsto \f A [x_1:x_2]^{\top}$.  This defines a group homomorphism $\mu:\aff\rightarrow\aut(\p^1)$.

\begin{theorem}\label{thm:good cand aff}
  Let $C$ and $D$ be curves as above. Then $\aut(\p^1)_{F_n,G_n}$ is a~good candidate set for $\aff$ where the inclusion map is given by the restriction $\iota=\mu|_{\aff_{C,D}}$.
\end{theorem}

\begin{proof}
  Any affine transformation between $C$ and $D$ must map the ideal points of $C$ to the ideal points of $D$. Hence its restriction to $\omega$ can be naturally viewed as a transformation from $\aut(\p^1)_{F_d,G_d}$.  Since $C$ and $D$ are irreducible there exist exactly $d$ intersections of each curve with $\omega$ and thus $\aut(\p^1)_{F_d,G_d}$ is finite whenever $d>2$. Thus in order to show that it is a good candidate set it remains to prove that $\iota$ is injective.  Assume a contradiction. Let there exists two different transformations $\phi_1,\phi_2\in\aff_{C,D}$ which are mapped to the same transformation in $\aut(\p_1)$. It is easy to see that in this case $\phi_2\circ\phi_1^{-1}$ is a translation or scaling in $\aff$. However the only irreducible algebraic curves invariant under these transformations are lines.
\end{proof}

Now, let be given $\phi \in\aut(\p^1)_{F_d,G_d}$, i.e., there exists a regular matrix $\f A$ and $\lambda\in\C^*$ such that $G_d(\f A (x_1,x_2)^{\top})=\lambda F_d(x_1,x_2)$.  In order to find its preimage given by $\overline{\f A}$, cf. \eqref{eq:matrix}, in $\aff_{C,D}$ under $\iota$ it is enough to compute $a_{00}$, $a_{10}$ and $a_{20}$  such that $G(\overline{\f A}(x_0,x_1,x_2)^{\top}) = \lambda F(x_0,x_1,x_2)$. This leads to a system of polynomial equations in $a_{00}$, $a_{10}$ and $a_{20}$.  Although it might seem to be complicated, we know that the system has no or exactly one solution depending on the existence of the preimage.  In addition, writing $G(\overline{\f A} (x_0,x_1,x_2)^{\top})={G'_d}(x_1,x_2)+G'_{d-1}(x_1,x_2)x_0+\cdots G'_0(x_1,x_2)x_0^d$ one  can show that the subsystem of equations corresponding to $G'_{d-1}(x_1,x_2)=\lambda F_{d-1}(x_1,x_2)$ is linear.  Moreover it can be easily seen that for curves in general position this system has a unique solution and thus we can avoid solving non-linear systems.

\begin{algorithm}[H]
\caption{Affine equivalences of planar curves.}\label{alg:aff curves} \algsetup{indent=2em}
\begin{algorithmic}[1]
\REQUIRE Curves $C: F(x_0,x_1,x_2)=0$ and $D: G(x_0,x_1,x_2)=0$, both of degree $d$.

\STATE Compute $\aut(\p^1)_{F_d,G_d}$, cf. Algorithm~\ref{alg:finite poly}.

\STATE For each $\f A\in\aut(\p^1)_{F_d,G_d}$ construct a matrix $\overline{\f A}$, cf.\eqref{eq:matrix}, with $ a_{11}, a_{12}, a_{21}, a_{22}$ given by $\f A$ and $a_{00}, a_{10}, a_{20}$ as free parameters.

\STATE Set $G'(x_0,x_1,x_2)=G(\overline{\f A} (x_0,x_1,x_2)^{\top})$.

 \IF{the linear system corresponding to $G'_{d-1}(x_1,x_2) = F_{d-1}(x_1,x_2)$ has a solution}
   \STATE
     include map corresponding to the matrix $\overline{\f A}$, i.e., the solution $a_{00}, a_{10}, a_{20}$ together with $\f A$, into $\aff_{C,D}$.
   \ENDIF

\ENSURE $\aff_{C,D}$.
\end{algorithmic}
\end{algorithm}

\begin{example}\rm
Consider two algebraic curves $C$ and $D$ of degree $5$ given by the forms
\begin{multline}
F=-23 x_0^5-109 x_1 x_0^4-7 x_2 x_0^4-179 x_1^2 x_0^3+5 x_2^2 x_0^3-54 x_1 x_2 x_0^3-22 x_1^3
   x_0^2
\\
-4 x_2^3 x_0^2-6 x_1 x_2^2 x_0^2-40 x_1^2 x_2 x_0^2+70 x_1^4 x_0-2 x_2^4 x_0-12 x_1
   x_2^3 x_0-28 x_1^2 x_2^2 x_0
\\
-28 x_1^3 x_2 x_0+49 x_1^5+x_2^5+5 x_1 x_2^4+2 x_1^2 x_2^3-6
   x_1^3 x_2^2+13 x_1^4 x_2.
\end{multline}
and
\begin{multline}
G= x_0^5-2 x_1 x_0^4+x_2 x_0^4+x_1^2 x_0^3+9 x_2^2 x_0^3+4 x_1 x_2 x_0^3+10 x_1^3 x_0^2+42 x_2^3
   x_0^2
\\
+65 x_1 x_2^2 x_0^2+41 x_1^2 x_2 x_0^2+10 x_1^4 x_0+63 x_2^4 x_0+139 x_1 x_2^3 x_0+128
   x_1^2 x_2^2 x_0
\\
+57 x_1^3 x_2 x_0 +2 x_1^5+31 x_2^5+87 x_1 x_2^4+102 x_1^2 x_2^3+61 x_1^3
   x_2^2+18 x_1^4 x_2.
\end{multline}
By computing the automorphisms $\aut(\p^1)_{F_d,G_d}$ of the forms
 \begin{equation}
F_d = 49 x_1^5+13 x_2 x_1^4-6 x_2^2 x_1^3+2 x_2^3 x_1^2+5 x_2^4 x_1+x_2^5,
\end{equation}
\begin{equation}
G_d = 2 x_1^5+18 x_2 x_1^4+61 x_2^2 x_1^3+102 x_2^3 x_1^2+87 x_2^4 x_1+31 x_2^5
\end{equation}
we arrive at
\begin{equation}
\f A = \left(
\begin{array}{cc}
 -4 & -2 \\
 3 & 1 \\
\end{array}
\right)
\end{equation}
When setting $G'(x_0,x_1,x_2)=G(\overline{\f A}(x_0,x_1,x_2)^{\top})$, the condition
$G'_{d-1}(x_1,x_2) = F_{d-1}(x_1,x_2)$ leads to the system of linear equations 
\begin{equation}
\begin{array}{rclclcr}
 a_{00}&-&5 a_{10}&-&5 a_{20}&=&-2, \\
 -6 a_{00}&-&28 a_{10}&-&36 a_{20}&=&-12, \\
 -4 a_{00}&-&18 a_{10}&-&30 a_{20}&=&-28, \\
 62 a_{00}&+&44 a_{10}&+&76 a_{20}&=&-28,\\
 139 a_{00}&+&103 a_{10}&+&219 a_{20}&=&70, \\
\end{array}
\end{equation}
which has the following solution
\begin{equation}
a_{00} =  -2,\quad a_{10} =  -3,\quad a_{20} = 3.
\end{equation}
Altogether we obtain the resulting transformation described by the matrix
\begin{equation}
\overline{\f A} = \left(
\begin{array}{ccc}
 -2 & 0 & 0 \\
 -3 & -4 & -2 \\
 3 & 3 & 1 \\
\end{array}
\right)
\end{equation}
which maps $C$ to $D$.
\end{example} 

\subsection{Similarities and symmetries of surfaces}

In this section we solve the problem of the detection of direct similarities of algebraic surfaces $R$ and $S$, i.e, finding $\simm_{R,S}$, by the computation of projective equivalences between finite sets of points in $\PP^1$. Again we will employ Algorithm \ref{alg:finite poly}.

The real algebraic surfaces $R,S$  are given as real solutions of polynomial equations $F(x_0,x_1,x_2,x_3)=0$ and $G(x_0,x_1,x_2,x_3)=0$, where $F,G$ are typically defined over $\Q$ or its finite extension. We make a~natural assumption that the polynomials $F,G$ are irreducible over $\C$ and that $\dim_\R R=\dim_\R S=2$. Since the degree of a surface is a projective invariant, we assume that both surfaces have the same degree. In addition we have (see e.g. \cite{AlHe16} for a more detailed analysis):

\begin{proposition}
If $R,S$ are not both cylinders, cones or surfaces of revolution then $\simm_{R,S}$ is finite.
\end{proposition}

There exist efficient algorithms for recognizing surfaces invariant under translations (cylinders), scalings (cones) and a~one parameter set of rotations (surfaces of revolution) and we assume that $R$ and $S$ are not surfaces of these types.

%Denote $F(x_0,\dots,x_3)=x_0^{\deg f}f(x_1/x_0,x_2/x_0,x_3,x_0)$ the homogenization of the polynomial $f$ and let $R_{\C}$ be its zero locus in $\p^3_{\C}$. Analogously we have $G(x_0,\dots,x_3)$ for $S_{\C}$.

Write $R_{\C}$ and $S_{\C}$ for  the zero locus of $F$ and $G$ in $\p^3_{\C}$.
The group $\simm$ acts naturally on $\PP^3_\C$ and it allows to consider also $\simm_{R_{\C},S_{\C}}$. Since any $\phi\in\simm_{R,S}$ maps the~real points of $R$ to the~real points of $S$ we obtain $\simm_{R_{\C},S_{\C}}\subset\simm_{R,S}$. Nevertheless the inclusion may be proper. The following lemma legitimizes our assumptions on the surfaces.

\begin{lemma}\label{lem C vs R}
 Let $R$ and $S$ be irreducible surfaces with $\dim_\R R=\dim_\R S=2$ then $\simm_{R_{\C},S_{\C}}=\simm_{R,S}$.
\end{lemma}
\begin{proof}
This follows directly from the fact that there is (up to a~constant factor) a unique irreducible polynomial vanishing exactly on $R_\R$ whenever $\dim_\R R=2$.
\end{proof}

%Our method relies on replacing the surfaces $R,S\subset\E_\R^3$ by a finite sets of points $\Lambda_S, \Lambda_R$ in $\PP^1_\C$  such that $\simm_{R,S}$ is isomorphic to a subgroup of $\aut(\p^1)_{\Lambda_S,\Lambda_R}$, together with the rule assigning to each transformation in this subgroup the correct isometry of the surfaces, see the following lemma.

Lemma~\ref{lem C vs R}  enables us to replace the~real surface $R$ by the complex one. Let us write
\begin{equation}
  F(0,x_1,x_2,x_3)=(x_1^2+x_2^2+x_3^2)^k\widetilde{F}(x_1,x_2,x_3),
\end{equation}
where $\sum x_i^2$ does not divide $\widetilde{F}$.
Then we set $R_\Omega$ to be the intersection of the absolute conic $\Omega$ and the curve $\widetilde{F}=0$ in the plane at infinity, together with the intersection multiplicities. Hence set theoretically  $R_\Omega=\overline{(R_\C\cap\omega)\backslash\Omega}\cap\Omega$ and it is a finite subset of $\Omega$. Since every similarity $\phi$ preserves the~absolute conic  we have $\phi(R_\Omega)=R_\Omega$. The conic $\Omega$ is a smooth rational curve and thus there exists an~isomorphism $\mu:\PP^1_\C\rightarrow\Omega$, for example it can be given by
\begin{equation}
   [s:t]\mapsto [0:2st:s^2-t^2:\I(s^2+t^2)].
\end{equation}
The pull-back $\mu^*\widetilde{F}$ is a form of degree $2\deg R$ on $\PP^1$ such that its zero-set is exactly the pre-image of the set $R_\Omega$ together with the~multiplicities. Analogously, we obtain $\mu^*\widetilde{G}$ for the surface $S$.

\begin{theorem}
Let $R$ and $S$ be surfaces as above. If $\widetilde{F}$ and $\widetilde{G}$ are not constants,  then $\aut(\p^1)_{\mu^*\widetilde{F},\mu^*\widetilde{G}}$ is a~good candidate set for $\simm_{R,S}$ where the inclusion map is given by the restriction $\iota=\mu|_{\simm_{R,S}}$.
\end{theorem}

\begin{proof}
The proof is analogous to the proof of Theorem~\ref{thm:good cand aff} with the specification that $\iota$ is not injective only for cylinders and cones, which were excluded from our considerations.
\end{proof}

%Algorithm~\ref{alg:finite poly} provides us the set of projective transformations mapping the roots of $\mu^*\widetilde{F}$ onto the roots of $\mu^*\widetilde{G}$. Let $\phi\in\aut(\p^1)_{\mu^*\widetilde{F},\mu^*\widetilde{G}}$, then $\psi=\mu\circ\phi\circ\mu^{-1}$ is an automorphism of $\Omega$ mapping the set $R_\Omega$ to $S_\Omega$. Conversely any such an automorphism can be obtained in this way.  Obviously any similarity in $\simm_{R,S}$ induces a transformation from $\aut(\p^1)_{\mu^*\widetilde{F},\mu^*\widetilde{G}}$. Nevertheless the converse is not true in general.
%
%
%For every $\phi\in\trans_{\mu^*F'}$ we have to find the possible similarity and check whether the surface is symmetric w.r.t. this similarity. Thus we call $\trans_{\mu^*F'}$ the~\emph{candidate group}.

Let $\phi\in\aut(\p^1)_{\mu^*\widetilde{F},\mu^*\widetilde{G}}$, then $\psi=\mu\circ\phi\circ\mu^{-1}$ is an automorphism of $\Omega$ mapping the set $R_\Omega$ to $S_\Omega$.
Since $\psi$ is an automorphism of a conic section in $\PP^2$ there exists a projective transformation $\Psi$ of $\PP^2$ preserving $\Omega$ such that $\Psi|_{\Omega}=\psi$. Write $\f A$ for a matrix representing $\Psi$. We can determine the remaining coefficients $a_{00}, \ldots, a_{03}$ of the matrix $\overline{\f A}$ of the possible similarity. In particular we solve the system of linear equations analogously as in Section \ref{Sub:affine}. The method is summarized in the Algorithm~\ref{alg:sim surfaces}.

%But projective transformations  preserving $\Omega$ can be identified with the group of rotations in $\R^3$.
%
%Hence for each $\phi\in\trans_{\mu^*F'}$ we immediately obtain a candidate for a symmetry of original surface in the form $\f x\mapsto\f A \cdot\f x+\f b$, where $\f A\in\SO$ and $\f b=(b_1,b_2,b_3)$ is an~unknown vector. In order to determine $\f b$ we can simply compare the $\deg f-1$ part of polynomials $f(\f x)$ and $f(\f A \cdot\f x+\f b)$. This leads to a system of linear equations in $(b_1,b_2,b_3)$ which can be easily solved.

\begin{algorithm}[H]
\caption{Similarities of algebraic surfaces.}\label{alg:sim surfaces}\algsetup{indent=2em}
\begin{algorithmic}[1]
\REQUIRE Surfaces $R: F(x_0,x_1,x_2,x_3)=0$ and $S: G(x_0,x_1,x_2,x_3)=0$, both of degree $d$.

\STATE Compute $\aut(\p^1)_{\mu^*\widetilde{F},\mu^*\widetilde{G}}$, cf. Algorithm~\ref{alg:finite poly}.

\STATE For each $\phi\in\aut(\p^1)_{\mu^*\widetilde{F},\mu^*\widetilde{G}}$ construct a projective transformation $\Psi$ (a matrix $\f A$) of $\PP^2$ preserving $\Omega$ such that $\Psi|_{\Omega}= \mu\circ\phi\circ\mu^{-1}$.

\STATE For each such $\f A$ construct a matrix $\overline{\f A}$, cf.\eqref{eq:matrix form}, given by $\f A$, $\widehat{\f a}=(0,\ldots,0)$ and $a_{00}, \ldots, a_{30}$ as free parameters.

\STATE Set $G'(x_0,x_1,x_2,x_3)=G(\overline{\f A} (x_0,x_1,x_2,x_3)^{\top})$.

 \IF{the linear system corresponding to $G'_{d-1}(x_1,x_2,x_3) = F_{d-1}(x_1,x_2,x_3)$ has a solution}
   \STATE
     include map corresponding to the matrix $\overline{\f A}$ (i.e., the solution $a_{00}, \ldots, a_{30}$ together with~$\f A$) into $\simm_{S,R}$.
   \ENDIF

\ENSURE $\simm_{S,R}$.
\end{algorithmic}
\end{algorithm}

Let us remark, that a self-similarity of an algebraic surface is immediately an isometry, hence  $\simm_{S}$ determines a group of symmetries. Since we assume that the surface is neither a cylinder, a cone or a surface of revolution, we know that $\simm_{S}$ is trivial, cyclic, dihedral or a group of symmetries of a platonic solid, cf. Proposition \ref{prp:platonic}.

We conclude this section by visualising the candidate group of the surface. This is achieved by the following construction of the map between $\Omega$  and the unit sphere $\mathcal{S}^2 \subset \R^3$. The point from $\p^3$ is contained in $\Omega$ if and only if it has coordinates $[0:\f p]$ such that
$\re(\f p)\cdot\im(\f p)=0$ and $|\re(\f p)|=|\im(\f p)|$. Consider a mapping $\gamma:\Omega\rightarrow\mathcal{S}^2$ defined by
\begin{equation}
  \gamma:\,\f p\mapsto\frac{\I}{\,\f p\cdot\overline{\f p}}(\f p\times\overline{\f p})=\frac{\re(\f p)\times\im(\f p)}{|\re(\f p)|\,|\im(\f p)|}.
\end{equation}
A rotation $\phi$ can be  viewed simultaneously as a~mapping $\phi:\Omega\rightarrow\Omega$ or $\phi:\mathcal{S}^2\rightarrow\mathcal{S}^2$. Since for a regular matrix $\f M$ and two vectors $\f a$, $\f b$ it holds $(\f M\f a)\times(\f M\f b)=\det \f M \cdot \f M^{-\top}(\f a\times\f b)$, the mapping $\gamma$ leads to a commutative diagram

\begin{equation}\label{eq diagram}
\begin{array}{c}
\xymatrix{
 \Omega  \ar[r]^{\phi}\ar[d]_{\gamma} & \Omega\ar[d]^{\gamma}\\
 \mathcal{S}^2\ar[r]^{\phi} & \mathcal{S}^2
}
\end{array}.
\end{equation}
For the sake of simplicity we assume that the surface $S$ intersects $\Omega$ with multiplicities one (otherwise we should consider also the multiplicities of the intersections and modify the approach accordingly). Denote $\Lambda:=\gamma(S_\Omega)$ the image of the intersections of the surface $S$ with the absolute conic on the unit sphere in $\R^3$. Let $0<d_1<\cdots <d_k<2$ be all the possible distances between points of $\Lambda$ (except of antipodal points) and write
\begin{equation}
\Lambda_i:=\left\{\f p\in\Lambda:\ \exists\f q\in\Lambda\ \text{such that}\ |\f p\,\f q|=d_i\right\},
\end{equation}
for $i\in\{1,\dots,k\}$. Since isometries preserve the distances we see that any isometry preserving $\Lambda$ must preserve each $\Lambda_i$. Conversely if every $\Lambda_i$ is preserved by some isometry, then the same is true for their union $\Lambda$.  And thus we have the following formula for the candidate group

\begin{equation}
  \aut(\p^1)_{\mu^*\widetilde{F}}\cong \bigcap_{i=1}^k\simm_{\Lambda_i}
\end{equation}

The procedure is illustrated in Fig.~\ref{fig eistuete} when we arrive at the  case $\simm_S = \bigcap_{i=1}^k\simm_{\Lambda_i}\simeq \mathcal{D}_4$. Furthermore, in Fig.~\ref{fig chair} we see a situation where $\simm_S \simeq {\cal T}$ is a proper subgroup of the candidate group  $\bigcap_{i=1}^k\simm_{\Lambda_i}\simeq {\cal O}$.

\begin{figure}[ht]
\begin{center}

  \includegraphics[width=0.22\textwidth]{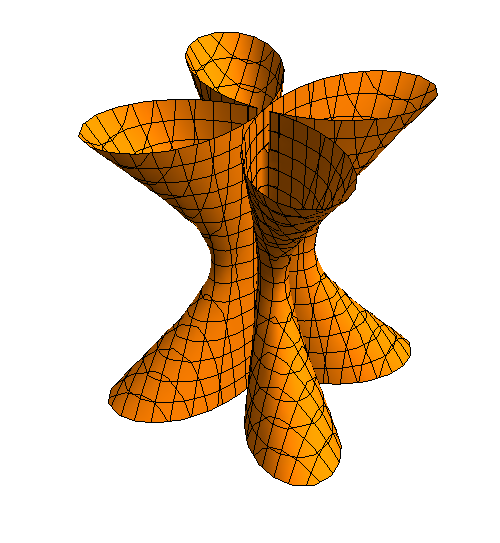}
  \includegraphics[width=0.22\textwidth]{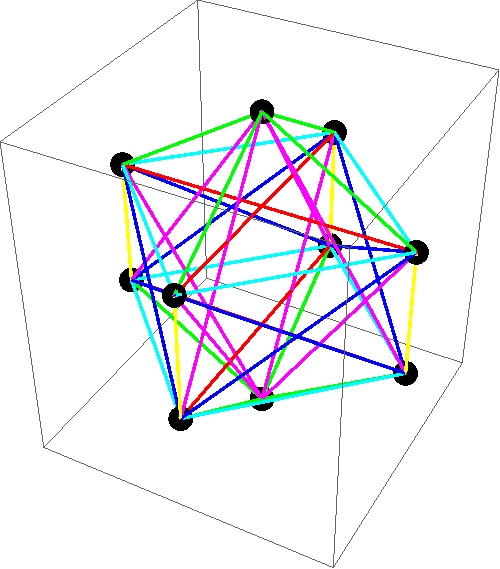}
  \includegraphics[width=0.22\textwidth]{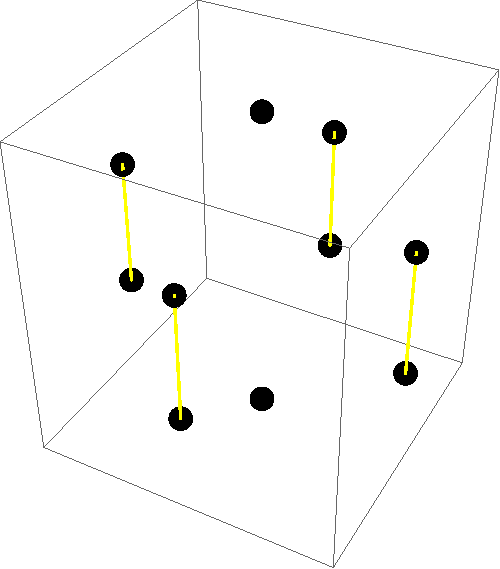}
  \includegraphics[width=0.22\textwidth]{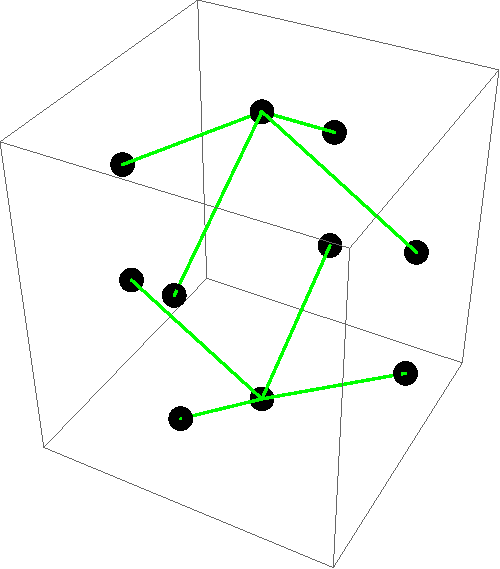}\\

  \includegraphics[width=0.22\textwidth]{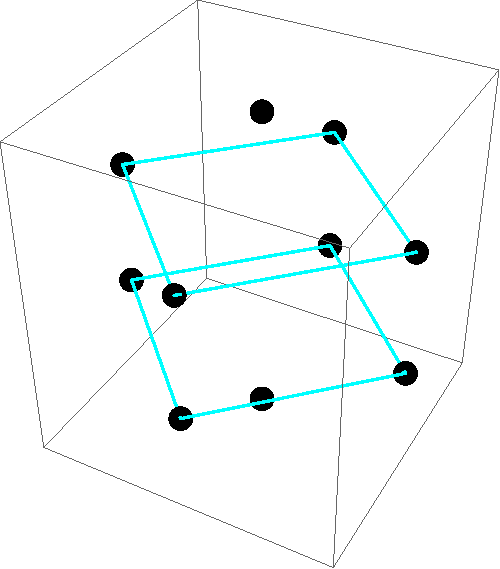}
  \includegraphics[width=0.22\textwidth]{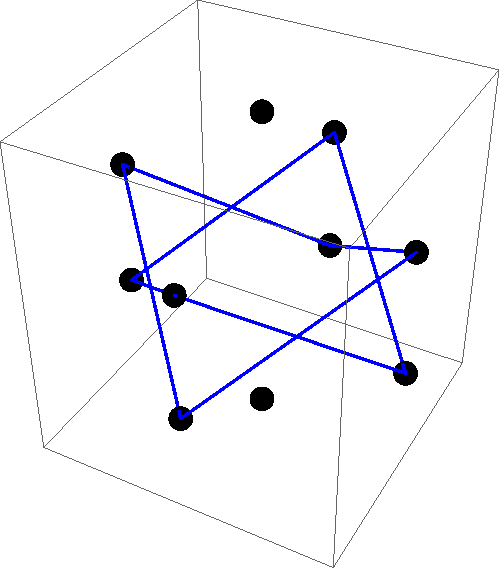}
  \includegraphics[width=0.22\textwidth]{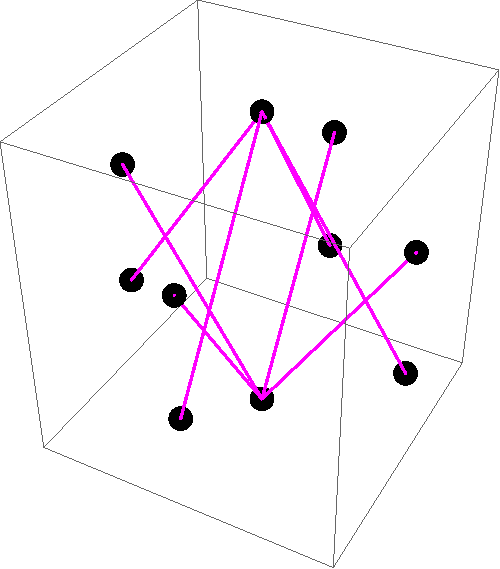}
  \includegraphics[width=0.22\textwidth]{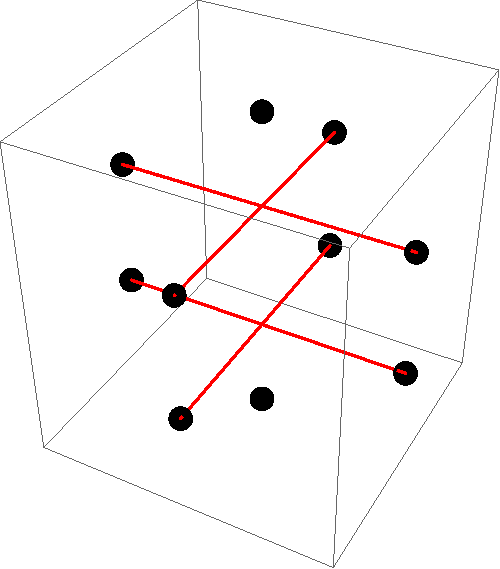}
\begin{minipage}{0.9\textwidth}
\caption{The surface $(x^2 + y^2)^3 - 4x^2y^2(z^2 + 1)$ with $\simm_S \simeq {\cal D}_4$. \label{fig eistuete}}
\end{minipage}
\end{center}
\end{figure}

\begin{figure}[ht]
\begin{center}

  \includegraphics[width=0.22\textwidth]{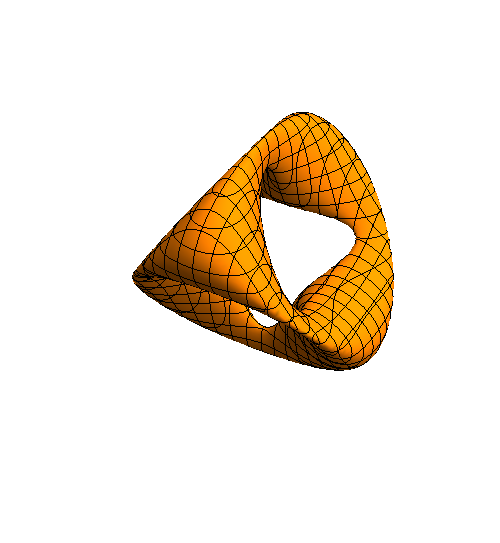}
  \includegraphics[width=0.22\textwidth]{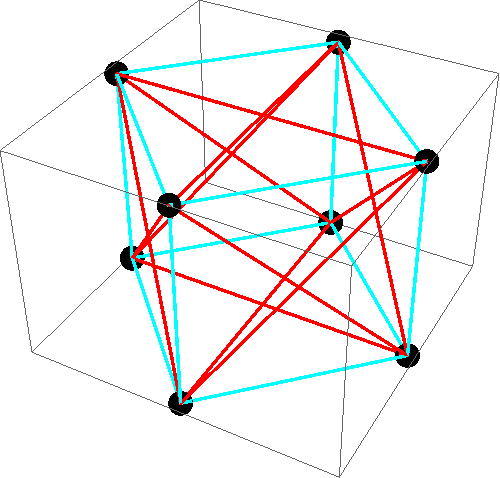}
  \includegraphics[width=0.22\textwidth]{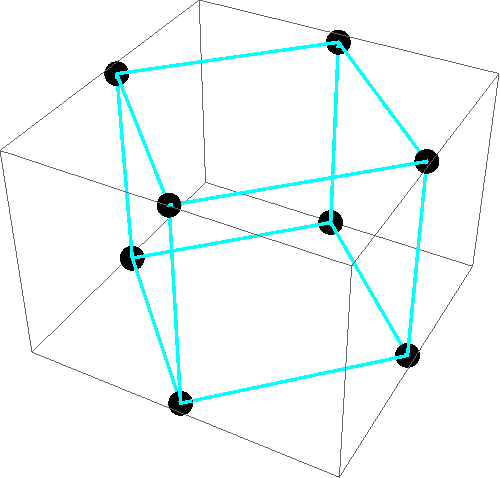}
  \includegraphics[width=0.22\textwidth]{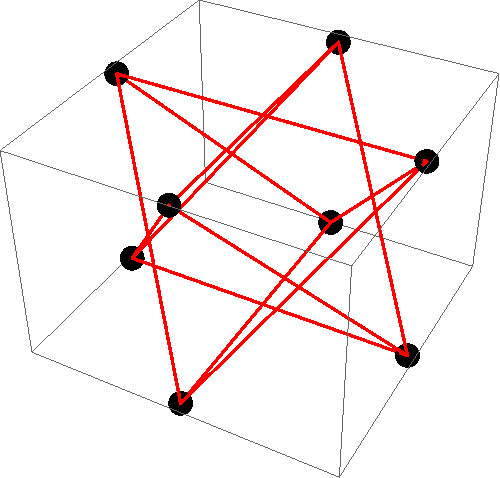}

\begin{minipage}{0.9\textwidth}
\caption{{\em Chair surface} with the tetrahedral symmetry $\simm_S \simeq {\cal T}$. \label{fig chair}}
\end{minipage}
\end{center}
\end{figure}

\section{Conclusion}

An identification of  a suitable class of equivalencies in  given geometric data is a topic interesting not only from the theoretical but also from the practical point of view. From this reason, computing projective equivalences of distinguished algebraic varieties has become an active research area and various situations are incessantly investigated. And as direct computations (although they can be formulated easily) are getting quite complicated even for trivial inputs, alternative efficient approaches are still required and investigated.

In this paper, we studied several situations that  can be transformed to determining equivalences of finite subsets of the projective line. This makes the designed method computationally  suitable e.g. for finding  projective equivalences of rational curves, determining projective equivalences of rational ruled surfaces, detecting affine transformations between planar algebraic curves, and computing similarities between two implicitly given algebraic surfaces.  The designed algorithms were implemented in the CAS Mathematica and their functionality was documented on several examples.

\section*{Acknowledgments}
\noindent
The authors were supported by the project LO1506 of the Czech Ministry of Education, Youth and Sports.
%We thank to all referees for their valuable comments, which helped us significantly to improve the paper.

%
%  Bibliography. Follow the usual conventions.
%

%\bibliographystyle{siam}
\bibliographystyle{elsarticle-harv}
\bibliography{bibliography,bibliographyML,new}

\end{document}